\documentclass[a4paper,10pt,eqno]{article}
\usepackage{latexsym,amssymb,amsfonts,amsmath}
\usepackage{amsthm}
\usepackage[dvips]{graphicx}
\newtheorem{thm}{Theorem}[section]
\newtheorem*{thm*}{Theorem}
\newtheorem{lem}[thm]{Lemma}

\newtheorem{pro}[thm]{Proposition}

\newcommand{\RM}{\mathbb{R}}
\newcommand{\ZM}{\mathbb{Z}}

\newcommand{\CM}{\mathbb{C}}

\newcommand{\HM}{\mathbb{H}}
\newcommand{\ket}[1]{|#1\rangle}
\newcommand{\bra}[1]{\langle #1|}
\newcommand{\ol}[1]{\overline{#1}}

\title{{\Large {\bf Probability distributions and weak limit theorems of quaternionic quantum walks in one dimension}}}
\author{
{\small Kei Saito\footnote{saito-kei-nb@ynu.jp},}\\
{\scriptsize  Department of Applied Mathematics, Faculty of Engineering, Yokohama National University}\\
{\scriptsize \footnotesize\it 79-5 Tokiwadai, Hodogaya, Yokohama, 240-8501, Japan}\\
}
\date{\empty}
\begin{document}
\maketitle

\par\noindent
\begin{small}
\par\noindent
{\bf Abstract}. The discrete-time quantum walk (QW) is determined by a unitary matrix whose component is complex number.
Konno (2015) extended the QW to a walk whose component is quaternion.We call this model quaternionic quantum walk (QQW).
The probability distribution of a class of QQWs is the same as that of the QW. On the other hand, a numerical simulation suggests that the probability distribution of a QQW is different from the QW.
In this paper, we clarify the difference between the QQW and the QW by weak limit theorems for a class of QQWs.

\footnote[0]{
{\it Abbr. title:} Probability distributions and weak limit theorems of quaternionic quantum walks in one dimension
}
\end{small}

\setcounter{equation}{0}

\section{Introduction}
\label{Introduction}
The quantum walk (QW) is a quantum dynamics defined as a quantization of the classical random walks. The study of quantum walks has recently begun to attract to the concern of various research fields such as information science and quantum physics. Moreover QW is powerful method for developing new quantum algorithms and protocols. Especially the discrete-time QW on the one-dimensional lattice is largely investigated and proposed to some kinds of models. As a remarkable property for a class of QWs, the quantum walker has both properties staying at the starting position and spreading quadratically faster than classical random walker.

The quaternionic quantum walk (QQW) on the one-dimensional lattice is introduced by Konno \cite{KonnoQQW} as a natural quaternionic extension of QW. Konno, Mitsuhashi, and Sato \cite{MitsuhashiQQW1, MitsuhashiQQW2, MitsuhashiQQW3} studied some properties about the spectrum of QQW on some graphs. Both QQWs are defined by extending complex components of the unitary matrix which governs the dynamics of corresponding QW to quaternion components.
The present paper treats only QQW on the one-dimensional lattice whose detailed definition is given in Sect. 2.
We have a concern for the probability distribution of the QQW, and obtain the concrete formulation of the distribution for some cases of the QQW.
Our results show the probability distribution of QQWs which belong to Cases 1 to 4 including a QQW introduced as an example in \cite{KonnoQQW} has the same formulation of that of the QW. However, in general case, the formulation doesn't always correspond to that of the QW. For instance, a numerical simulation suggests that the probability distribution of a QQW is shaped by the superposition of some distributions (see Fig. 1). Moreover in the QW, Konno \cite{KonnoQW2} showed the support of range of limit density function is determined by the modulus of a component of unitary matrix which is called coin operator. In contrast, our results show this support of QQWs belonging to Case 5 is not determined by only the modulus of a component (see Fig. 2). 

\begin{figure}[hbt]
\begin{tabular}{cc}
\begin{minipage}[t]{.45\textwidth}
\centering
\includegraphics[width=70mm,keepaspectratio,clip]{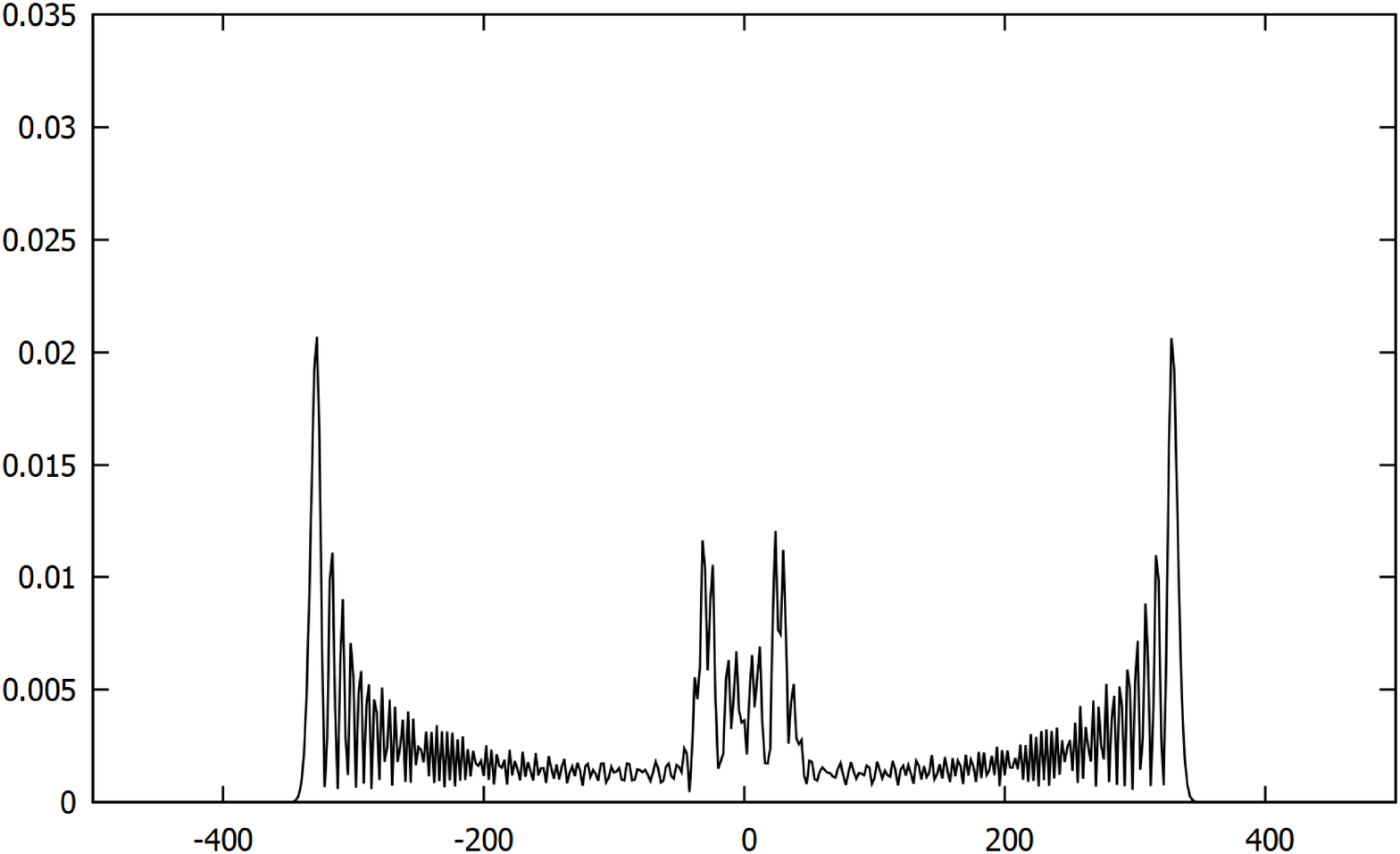}
\caption{A QQW whose distribution is given by a superposition of distributions.}
\end{minipage}&
\begin{minipage}[t]{.45\textwidth}
\centering
\includegraphics[width=70mm,keepaspectratio,clip]{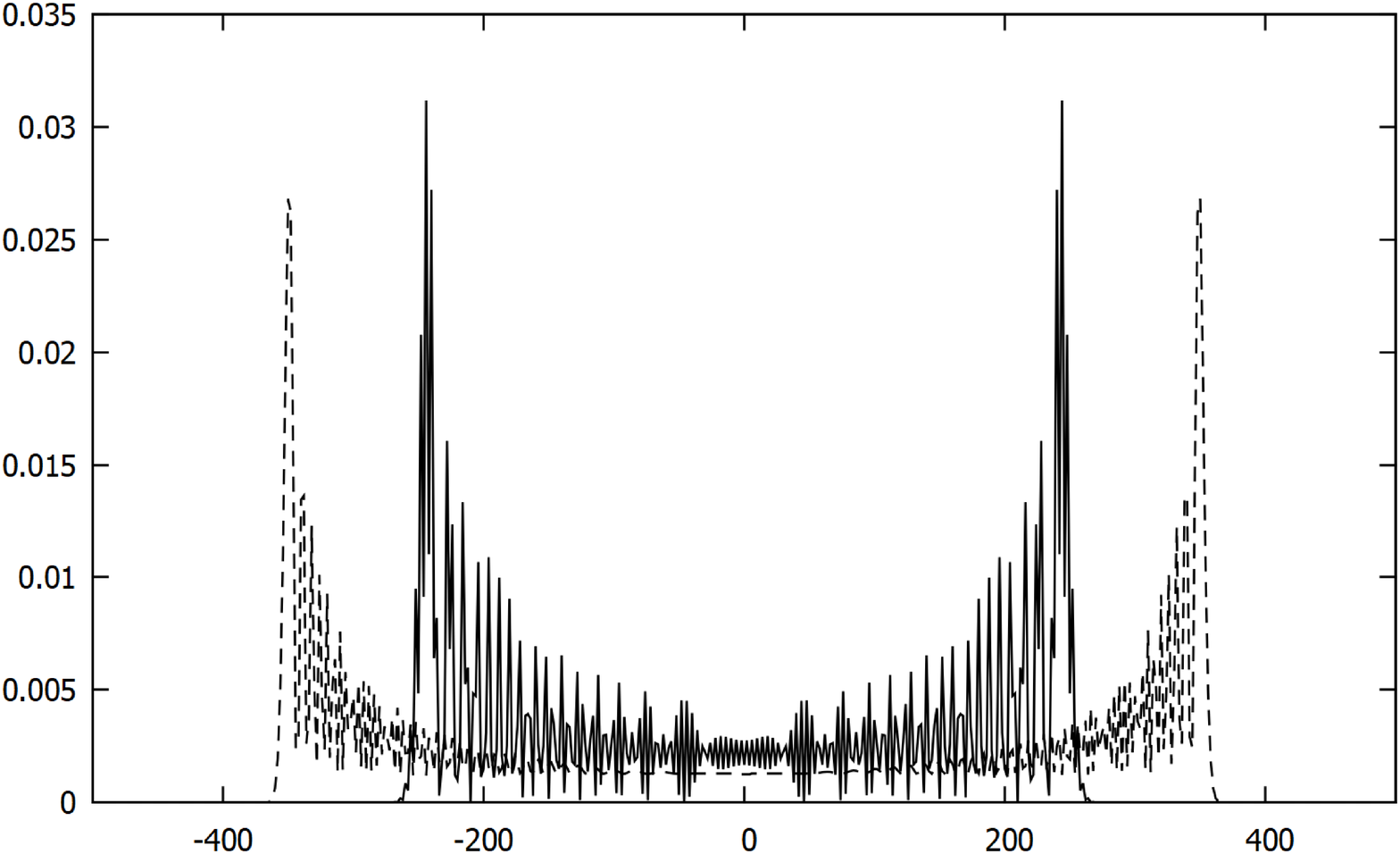}
\caption{Hadamard QW (dotted line) and a QQW of Case 5 (solid line). Both coin operators of the walks have same parameters by modulus. }
\end{minipage}
\end{tabular}
\end{figure}

The rest of the present paper is organized as follows. Some notations and definitions used in this paper are given in Sect. 2. Furthermore the detailed definition of QQWs is also introduced in this section. In Sect. 3, we show the formulations of the probability distribution of QQWs for Case 1 to 4 are the same as that of QWs. By considering the weak limit theorem, Sect. 4 presents a class of QQW whose behaviour is different from the QW.
The limit density function of this class is an extension of that of QW.

\section{Definitons and models}
\subsection{Quaternion}
The quaternion was introduced as an extension of the complex number by Hamilton in 1843. 
Let $\RM$, $\CM$ and $\HM$ be the set of real, complex numbers and quaternions, respectively.  
Then $x\in\HM$ is expressed as follows. 
\begin{align*}
x=x_0+x_1i+x_2j+x_3k\in \HM,
\end{align*}
\vspace{2mm}
where $x_0, x_1, x_2, x_3 \in \RM$ and, 
\begin{align*}
i^2 &= j^2 = k^2 = -1,
\\
ij & = -ji = k, \quad jk= -kj =i, \quad ki=-ik = j.
\end{align*}
Hence $\HM$ is noncommutative algebra. For $x=x_0+x_1i+x_2j+x_3k\in\HM$ ($x_0,x_1,x_2,x_3\in\RM$), let $\ol{x}$ be the conjugate of $x$ whose form is given by
\begin{align*}
\overline{x} = x_0-x_1i-x_2j-x_3k.
\end{align*}
Moreover a modulus of $x$ is
\begin{align*}
|x| = \sqrt{x\ol{x}} = \sqrt{\ol{x}x} = \sqrt{x_0^2+x_1^2+x_2^2+x_3^2}.
\end{align*}
Let $\mbox{\boldmath{M}}(n, \CM)$ and $\mbox{\boldmath{M}}(n, \HM)$ be the set of all $n \times n$ matrices with complex and quaternion components, respectively. 
${}^T\!A$  denotes the transpose of $A$. For $A=(a_{st}) \in \mbox{\boldmath{M}}(n, \HM)$, we put $\ol{A}= (\ol{a}_{st})$ and $A^{*} ={}^T\!\left(\ol{A}\,\right)$. 
As with the complex components, $A$ is a unitary matrix, if $AA^{*}= A^{*}A=I$, where $I$ is the identity matrix. 
Let $\mbox{\boldmath{U}}(n, \CM)$ and $\mbox{\boldmath{U}}(n, \HM)$ be the set of all $n \times n$ unitary matrices with complex and quaternionic components, respectively. Moreover we divide $x$ into the real part as $\Re(x) = x_0$, and the imaginary part  as $\Im(x)= x_1i+x_2j+x_3k$. For another expression, $x$ is uniquely expressed as follows;
\begin{align*}
x=x'+x''j\in\HM
\qquad
(x'=x_0+x_1i,\ x''=x_2+x_3i\in\CM).
\end{align*}
Here $x'$ and $x''$ are called simplex and perplex parts, respectively. 
Furthermore we define the mapping $\chi:\mbox{M}(n,{\HM}) \to \mbox{M}(2n,{\CM})$ by
\begin{align*}
&\chi(x)=
\begin{bmatrix}
x' & -x''\\
\ol{x''} & \ol{x'}
\end{bmatrix}\in\mbox{M}(2,\CM).
\end{align*}
We remark that, for $A=(a_{st})\in\mbox{M}(n,\HM)$, we define $\chi(A)=(\chi(a_{st}))\in\mbox{M}(2n,\CM)$. For example 
\begin{align*}
U=\begin{bmatrix}a & b \\ c &  d \end{bmatrix}\in\mbox{M}(2,\HM), \quad \chi(U)=\begin{bmatrix} \chi(a) & \chi(b) \\ \chi(c) & \chi(d) \end{bmatrix}\in\mbox{M}(4,\CM).
\end{align*}
Then the following relations hold.
\begin{lem}
For $A,B\in\mathrm{M}(n,\HM)$,
\begin{align*}
&(1)\quad\chi(aA)=a\chi(A)\quad(a\in\RM),\\
&(2)\quad\chi(A)\chi(B)=\chi(AB),\\
&(3)\quad\chi(A+B)=\chi(A)+\chi(B).
\end{align*}
\end{lem}
\subsection{QQW}
The QQW on $\ZM$ is determined by the unitary matrix $U\in\mbox{U}(2,\HM)$ which is called coin operator, where $\ZM$ is the set of integers.
For the QW, the component of the matrix is a complex number. On the other hand, the component is a quaternion in the QQW. The walker of QQW has two chiralities, left and right, corresponding to the direction of the motion. Then we adapt each chirality to the vector $\ket{L} ={}^T\begin{bmatrix}1 & 0\end{bmatrix}$ and $\ket{R} ={}^T\begin{bmatrix}0 & 1\end{bmatrix}$, where $L$ and $R$ refer to the left and right chirality states, respectively. Let the coin operator $U\in\mbox{U}(2,\HM)$ be
\begin{align*}
U = \begin{bmatrix}
a & b \\
c & d 
\end{bmatrix}\in\mbox{U}(2,\HM).
\end{align*}
Unitarity of $U$ gives
\begin{lem}
\label{unitary}
\begin{align*}
&|a|^2+|b|^2=
|c|^2+|d|^2=1,\quad
a\overline{c}+b\overline{d}=
\overline{a}b+\overline{c}d=0,\quad
|a|^2=|d|^2,\quad
|b|^2=|c|^2.
\end{align*}
\end{lem}
\noindent Furthermore we divide $U$ into two matrices, $P$ and $Q$ defined by
\begin{align*}
P = \begin{bmatrix}
a & b \\
0 & 0 \end{bmatrix} ,\quad
Q = \begin{bmatrix}
0 & 0 \\
c & d \end{bmatrix}.
\end{align*}
$P$ and $Q$ represent that the walker moves to the left and right, respectively. Then the evolution of the quaternion version amplitude on position $x$ at time $n$, 
$\Psi_n(x)=
{}^T\!\begin{bmatrix} \Psi_{n}^L(x) & \Psi_{n}^R(x)
\end{bmatrix}\in\HM^2$, is defined by
\begin{align*}
\Psi_{n+1}(x)=P \Psi_n(x+1)+Q \Psi_n(x-1).
\end{align*}
That is,
\begin{align*}
\begin{bmatrix}
\Psi_{n+1}^L(x)\\
\Psi_{n+1}^R(x)
\end{bmatrix}
 = 
\begin{bmatrix}
a\Psi_n^L(x+1)+b\Psi_n^R(x+1)\\
c\Psi_n^L(x-1)+d\Psi_n^R(x-1)
\end{bmatrix}.
\end{align*}
The probability that the walker $X_n$ exists on position $x$ at time $n$ is defined by $||\Psi_n(x)||^2=\Psi_n^*(x)\Psi_n(x)$.
In this paper, we treat the model starting from only the origin. Hence we put the initial state $\Psi_0(x)\in\HM^2$ as
\begin{align*}
\Psi_0(x)=
\delta_0(x)\begin{bmatrix} \alpha \\ \beta \end{bmatrix} \quad(x\in\ZM),
\end{align*}
with $\alpha,\beta\in\HM$ and $|\alpha|^2+|\beta|^2=1$. Here, Kronecker's delta $\delta_0(x)$ equals to $1$ if $x=0$, equals to $0$ otherwise.

\subsection{Fourier transform}
We use the Fourier transform $\hat{\Phi}_n(\theta)\ (\theta\in[0,2\pi))$ given by
\begin{align*}
\hat{\Phi}_n(\theta)=\sum_{x\in\ZM}e^{-i\theta x}\Phi_n(x),
\end{align*}
where $\Phi_n(x)\in\CM^4$ means the first column of $\chi(\Psi_n(x))$. In other word,
\begin{align*}
\Phi_n(x)=\chi(\Psi_n(x))
\begin{bmatrix}
1 \\ 0
\end{bmatrix}.
\end{align*}
By the inverse Fourier transform, we have
\begin{align*}
\Phi_n(x)=\int^{\pi}_{-\pi}e^{i\theta x}\hat{\Phi}_n(\theta)\frac{d\theta}{2\pi}.
\end{align*}
According to the definition of $\Phi_n(x)$ and QQW, we get the following important relations.
\begin{lem}
\label{PhiProb}
\begin{align*}
&(1)\quad P(X_n=x)=||\Phi_n(x)||^2.\\
&(2)\quad \Phi_{n+1}(x)=\chi(P)\Phi_{n}(x+1)+\chi(Q)\Phi_{n}(x-1).
\end{align*}
\end{lem}
Remark that Lemma 2.3 implies that the QQW is essentially equivalent to the corresponding 4-state QW \cite{QWMemo1, QWMemo2}. Let us define $U(\theta)\in\mbox{M}(4,\CM)$ by
\begin{align*}
U(\theta)=
\begin{bmatrix}
\ \, e^{i\theta}\chi(a) & \ \,e^{i\theta}\chi(b) \\
e^{-i\theta}\chi(c) & e^{-i\theta}\chi(d)
\end{bmatrix}
&=
\begin{bmatrix}
e^{i\theta} & 0 & 0 & 0 \\
0 & e^{i\theta} & 0 & 0 \\
0 & 0 & e^{-i\theta} & 0 \\
0 & 0 & 0 & e^{-i\theta}
\end{bmatrix}
\begin{bmatrix}
\chi(a) & \chi(b) \\
\chi(c) & \chi(d)
\end{bmatrix}
\\
&=
%\begin{bmatrix}
%e^{i\theta} & 0 & 0 & 0 \\
%0 & e^{i\theta} & 0 & 0 \\
%0 & 0 & e^{-i\theta} & 0 \\
%0 & 0 & 0 & e^{-i\theta}
%\end{bmatrix}
\begin{bmatrix}
e^{i\theta}a' & -e^{i\theta}a'' & e^{i\theta}b' & -e^{i\theta}b''\\
e^{i\theta}\ol{a''} & e^{i\theta}\ol{a'} & e^{i\theta}\ol{b''} & e^{i\theta}\ol{b'}\\
e^{-i\theta}c' & -e^{-i\theta}c'' & e^{-i\theta}d' & -e^{-i\theta}d''\\
e^{-i\theta}\ol{c''} & e^{-i\theta}\ol{c'} & e^{-i\theta}\ol{d''} & e^{-i\theta}\ol{d'}
\end{bmatrix}
.
\end{align*}
Then we formulate the evolution by
\begin{pro}
\label{TeProU}
\begin{align*}
\hat{\Phi}_{n}(\theta) \ =\  U(\theta)\hat{\Phi}_{n-1}(\theta)\ =\  U(\theta)^n\hat{\Phi}_0(\theta).
\end{align*}
\end{pro}
\begin{proof}
Using Lemma \ref{PhiProb} (2), we have
\begin{align*}
U(\theta)\hat{\Phi}_{n-1}(\theta)&=U(\theta)\sum_{x\in\ZM}e^{-i\theta x}\Phi_{n-1}(x)=\sum_{x\in\ZM}\left(e^{-i\theta(x-1)}\chi(P)+e^{-i\theta(x+1)}\chi(Q)\right)\Phi_{n-1}(x)\\
&=\sum_{\tilde{x}\in\ZM}e^{-i\theta\tilde{x}}\chi(P)\Phi_{n-1}(\tilde{x}+1)+\sum_{\tilde{x}\in\ZM}e^{-i\theta\tilde{x}}\chi(Q)\Phi_{n-1}(\tilde{x}-1)=\sum_{\tilde{x}\in\ZM}e^{-i\theta\tilde{x}}\Phi_{n}(\tilde{x})=\hat{\Phi}_{n}(\theta).
\end{align*}
\end{proof}
\noindent By Lemma \ref{PhiProb} and Proposition \ref{TeProU}, the probability distribution is expressed as
\begin{align*}
P(X_n=x)&=||\Phi_n(x)||^2=\Phi_n^*(x)\Phi_n(x)=\int^{\pi}_{-\pi}e^{-i\theta x}\hat{\Phi}_n^*(\theta)\frac{d\theta}{2\pi}
\int^{\pi}_{-\pi}e^{i\theta' x}\hat{\Phi}_n(\theta')\frac{d\theta'}{2\pi}
\\[+8pt]
&=\int^{\pi}_{-\pi}\int^{\pi}_{-\pi}e^{i(\theta'-\theta)x}\left(\hat{\Phi}_0^*(\theta)U^*(\theta)^n\right)
\left(U(\theta')^n\hat{\Phi}_0(\theta')\right)\frac{d\theta}{2\pi}\frac{d\theta'}{2\pi}.
\end{align*}
In this paper, we focus on the initial state $\Psi_0(x)=\delta_0(x){}^T\begin{bmatrix}\alpha & \beta \end{bmatrix}\in\HM^2$. 
Then $\Phi_0(x)$ equals to $\hat{\Phi}_0(\theta)$, and the form is given as
\begin{align*}
\chi(\Psi_0(x))=\delta_0(x)
\begin{bmatrix}
\alpha_0+\alpha_1i & -\alpha_2-\alpha_3i \\
\alpha_2-\alpha_3i & \ \ \,\alpha_0-\alpha_1i \\
\beta_0+\beta_1i & -\beta_2-\beta_3i \\
\beta_2-\beta_3i & \ \ \,\beta_0-\beta_1i \\
\end{bmatrix}
=\delta_0(x)
\begin{bmatrix}
\alpha' & -\alpha''\\
\ol{\alpha''} & \ol{\alpha'}\\
\beta' & -\beta''\\
\ol{\beta''} & \ol{\beta'}\\
\end{bmatrix}
,
\quad
\Phi_0(x)=\delta_0(x)
\begin{bmatrix}
\alpha_0+\alpha_1i\\
\alpha_2-\alpha_3i\\
\beta_0+\beta_1i\\
\beta_2-\beta_3i
\end{bmatrix}
=\delta_0(x)
\begin{bmatrix}
\alpha'\\
\ol{\alpha''}\\
\beta'\\
\ol{\beta''}
\end{bmatrix}
,
\end{align*}
where $\alpha=\alpha_0+\alpha_1i+\alpha_2j+\alpha_3k=\alpha'+\alpha''j,\quad \beta=\beta_0+\beta_1i+\beta_2j+\beta_3k=\beta'+\beta''j$.
\section{Probability distribution}
This section introduces probability distributions of some classes of QQWs on $\ZM$. As one of main results of this paper, we present
\begin{thm}
\label{sameasQW}
For the coin operator $U=\begin{bmatrix}a & b\\ c & d\end{bmatrix}\in{\rm U}(2, \HM)$ defined by the following four cases, the formulations of the probability distribution of QQW are the same as that of QW.
\begin{align*}
&Case\ 1:\quad b=c=0\\
&Case\ 2:\quad a=d=0\\
&Case\ 3:\quad a,d\in\RM,\quad b,c\in\HM\\
&Case\ 4:\quad a=a_0+a_1i,\quad b=b_2j+b_3k,\quad c=c_2j+c_3k,\quad d=d_0+d_1i,
\end{align*}
\end{thm}
\noindent Especially, if $abcd=0$, then the QQW becomes Case 1 or Case 2. Moreover the probability distribution of the QW on $\ZM$ was studied by Konno\cite{KonnoQW1} as follows.
\begin{thm}{\rm Konno\cite{KonnoQW1}}\hspace{0.5cm}\label{KonnoPD} For $U=\begin{bmatrix}a & b\\c & d\end{bmatrix}\in{\rm U}(2,\CM)$ and $t=1,2,...,\left[ n/2 \right]$， we have\\
(1)\quad If $b=c=0$, then
\begin{align*}
P(X_n=x)=\delta_{-n}(x)|\alpha|^2+\delta_n(x)|\beta|^2.
\end{align*}
(2)\quad If $a=d=0$, then
\begin{align*}
P(X_{2n}=x)=\delta_0(x),
\qquad
P(X_{2n+1}=x)=\delta_1(x)|\alpha|^2+\delta_{-1}(x)|\beta|^2.
\end{align*}
(3)\quad If $abcd\neq0$, then
\begin{flalign*}
P(X_n &= \pm(n-2t))
\\
=&\ a^{2(n-1)}\sum_{\gamma=1}^{t}\sum_{\delta=1}^{t}\left(-\frac{|b|^2}{|a|^2}\right)^{\gamma + \delta}
\binom{t-1}{\gamma -1}\binom{t-1}{\delta -1}\binom{n-t-1}{\gamma -1}\binom{n-t-1}{\delta -1}\\
&\times\frac{1}{\gamma\delta}
\Bigg[
\frac{n^2-(\gamma+\delta+2t)n+2t^2}{2}
+\frac{\gamma\delta}{|b|^2}\\
&\pm\frac{n-2t}{2}\left\{n(|a|^2-|b|^2)+\gamma+\delta \right\}(|\beta|^2-|\alpha|^2)
\pm \frac{n-2t}{|b|^2}\left( \gamma+\delta-2n|b|^2\right)\Re(\ol{\alpha}\ol{a}b\beta)
\Bigg].
\end{flalign*}
\end{thm}
\noindent In general, we remark that $P(X_n=x)$ with $x=\pm n$ is given by
\begin{align*}
\Psi_n(n)=Q^n\Psi_0(x),\quad
\Psi_n(-n)=P^n\Psi_0(x),
\end{align*}
and
\begin{align*}
Q^n=d^{n-1}Q,\quad P^n=a^{n-1}P.
\end{align*}
Thus
\begin{align*}
P(X_n=n)&=|a|^{2(n-1)}
\left(
|b|^2|\alpha|^2+|a|^2|\beta|^2-2\Re(\ol{\alpha}\ol{a}b\beta)
\right),\\
P(X_n=-n)&=|a|^{2(n-1)}
\left(
|a|^2|\alpha|^2+|b|^2|\beta|^2+2\Re(\ol{\alpha}\ol{a}b\beta)
\right).
\end{align*}
\subsection{Case 1 \label{simple}}
In Case 1, the coin operator $U$ is given by
\begin{align*}
&U=
\begin{bmatrix}
a & 0 \\
0 & d
\end{bmatrix}
\in{\rm U}(2,\HM).
\end{align*}
Then $P$ and $Q$ are
\begin{align*}
P=\begin{bmatrix}
a & 0 \\
0 & 0
\end{bmatrix},
\quad
Q=\begin{bmatrix}
0 & 0 \\
0 & d
\end{bmatrix},
\quad
PQ=QP=O,
\end{align*}
where $O$ means the zero matrix.
Here $PQ=QP=O$ implies the amplitude and the probability distribution as
\begin{align*}
\Psi_n(x)=
\delta_{-n}(x)a^n \begin{bmatrix} \alpha \\ 0     \end{bmatrix}
+
\delta_{n}(x)d^n \begin{bmatrix} 0      \\ \beta \end{bmatrix},\qquad
P(X_n=x)=\delta_{-n}(x)|\alpha|^2+\delta_n(x)|\beta|^2.
\end{align*}
Thus we have Theorem \ref{KonnoPD} (1).
\subsection{Case 2 \label{local}}
In Case 2, the coin operator $U$ is given by
\begin{align*}
&U=
\begin{bmatrix}
0 & b \\
c & 0
\end{bmatrix}
\in{\rm U}(2,\HM).
\end{align*}
Then $P$ and $Q$ are
\begin{align*}
P=\begin{bmatrix}
0 & b \\
0 & 0
\end{bmatrix}
,\quad
Q=\begin{bmatrix}
0 & 0 \\
c & 0
\end{bmatrix}
,\quad
P^2=Q^2=O,\quad
PQ=b\begin{bmatrix} c & 0 \\ 0 & 0 \end{bmatrix}, \quad
QP=c\begin{bmatrix} 0 & 0 \\ 0 & b \end{bmatrix}.
\end{align*}
As with the Case 1, $P^2=Q^2=O$ implies the amplitude and the probability distribution as
\begin{align*}
&\Psi_{2n}(x)=
\delta_0(x)\begin{bmatrix} (bc)^n\alpha \\ (cb)^n\beta     \end{bmatrix},
\hspace{1cm}
\Psi_{2n+1}(x)=
\delta_1(x)\begin{bmatrix} 0 \\ c(bc)^n\alpha     \end{bmatrix}
+\delta_{-1}(x)\begin{bmatrix} b(cb)^n\beta \\ 0     \end{bmatrix},
\\[+10pt]
&P(X_{2n}=x)=\delta_0(x),
\hspace{1.9cm}
P(X_{2n+1}=x)=\delta_1(x)|\alpha|^2+\delta_{-1}(x)|\beta|^2. 
\end{align*}
So Theorem \ref{KonnoPD} (2) is obtained.
\subsection{Case 3 \label{a&d real}}
In Case 3, the coin operator $U$ is given by
\begin{align*}
&U=
\begin{bmatrix}
a & b \\
c & d
\end{bmatrix}
\in{\rm U}(2,\HM)
\quad
a,d\in\RM
\quad
b,c\in\HM.
\end{align*}
By Lemma \ref{unitary}, we see $d=\pm a$ and $c=\mp\overline{b}$. Then $U$, $P$ and $Q$ are expressed as
\begin{align*}
U=\begin{bmatrix}
\ \ a & \ \ b \\[+5pt]
\mp\overline{b} & \pm a
\end{bmatrix},
\quad
P=\begin{bmatrix}a & b \\ 0 & 0\end{bmatrix},\
\quad
Q=\begin{bmatrix} \ \ 0 & \ \ 0 \\ \mp\overline{b} & \pm a\end{bmatrix}.
\end{align*}
Here we note that $P$ and $Q$ are commutative for their own components $a$, $b$, and $\mp\ol{b}$. Therefore the coin operator for this class is treated as the same as QW. Then we use the following result.
\begin{thm} {\rm Konno\cite{KonnoQW1}}\quad
For $U=\begin{bmatrix} a & b \\ c & d \end{bmatrix}\in{\rm U}(2,\CM)$, $\Xi_n(l,m)$ is determined as follows.
\begin{align*}
&\Xi_n(l,m)=a^l d^m \sum_{\gamma=1}^{l\land m}\left(-\frac{|b|^2}{|a|^2}\right)^\gamma\binom{l-1}{\gamma-1}\binom{m-1}{\gamma-1}\times \frac{1}{\gamma}
\begin{bmatrix}
l& \displaystyle \frac{bcl+(ad-bc)\gamma}{ac}\\
\displaystyle \frac{bcm+(ad-bc)\gamma}{bd} & m
\end{bmatrix}
\end{align*}
with $l\land m=\min\{l,m\}$. 
\label{PQRS}
\end{thm}
\noindent Here $\Xi_n(l,m)$ means the sum of possible paths of the walker moved $l$-step left and $m$-step right at time $n$. For instance,
\begin{align*}
\Xi_4(1,3)=PQ^3+QPQ^2+Q^2PQ+Q^2P.
\end{align*}
We remark that $l$ and $m$ satisfy $m+l=n$ and $m-l=x$. By the definition, $\Psi_n(x)$ is expressed by $\Xi_n(l,m)$ as
\begin{align*}
\Psi_n(x)\ =\ \Xi_n(l,m)\Psi_0(x)\ =\ \Xi_n(l,m)\begin{bmatrix} \alpha \\ \beta \end{bmatrix}.
\end{align*}
By Theorem \ref{PQRS}, the sum of possible paths for this case is given by
\begin{align*}
&\Xi_n(l,m)=(\pm 1)^m a^n \sum_{\gamma=1}^{l\land m}\left(-\frac{|b|^2}{|a|^2}\right)^\gamma\binom{l-1}{\gamma-1}\binom{m-1}{\gamma-1}\frac{1}{\gamma}\times 
\begin{bmatrix}
l& \displaystyle \frac{|b|^2l-\gamma}{a|b|^2}b\\
\displaystyle \frac{-|b|^2m+\gamma}{a|b|^2}\overline{b} & m
\end{bmatrix}.
\end{align*}
Then we have
\begin{align*}
&\Psi_n(x)=(\pm 1)^m a^n \sum_{\gamma=1}^{l\land m}\left(-\frac{|b|^2}{|a|^2}\right)^\gamma\binom{l-1}{\gamma-1}\binom{m-1}{\gamma-1}\times \frac{1}{a|b^2|\gamma}
\begin{bmatrix}
a|b|^2l\alpha-(\gamma-|b|^2l)b\beta \\
a|b|^2m\beta+(\gamma-|b|^2m)\overline{b}\alpha
\end{bmatrix}.
\end{align*}
We obtain the probability distribution as the same as that of QW by computing $||\Psi_n(x)||^2$.
\subsection{Case 4 \label{cross}}
In Case 4, the coin operator $U$ is given by
\begin{align*}
&U=
\begin{bmatrix}
a & b \\
c & d
\end{bmatrix}
\in{\rm U}(2,\HM),
\quad
a=a_0+a_1i,\quad b=b_2j+b_3k,\quad c=c_2j+c_3k,\quad d=d_0+d_1i.
\end{align*}
Each parameter has only the simplex or perplex part. Specifically, $a=a'$, $d=d'$, $b=b''j$, and $c=c''j$. 
Then $\chi(P)$ and $\chi(Q)$ are
\begin{align*}
\chi(P)=\begin{bmatrix}
a' & \ 0 & \ 0 & -b''\, \\
0 & \ \overline{a'} & \ \ol{b''} & \ 0 \\
0 & \ 0 & \ 0 & \ 0 \\
0 & \ 0 & \ 0 & \ 0
\end{bmatrix},
\quad
\chi(Q)=\begin{bmatrix}
0 & \ 0 & 0 & 0 \\
0 & \ 0 & 0 & 0 \\
0 & -c'' & d' & 0 \\
\ol{c''} & \ 0 & 0 & \overline{d'} 
\end{bmatrix}.
\end{align*}
\vspace{1mm}
Furthermore we divide these matrices into
\begin{align*}
P_1=\begin{bmatrix}
a' & \ 0 & \ 0 & -b''\, \\
0 & \ 0 & \ 0 & \ 0 \\
0 & \ 0 & \ 0 & \ 0 \\
0 & \ 0 & \ 0 & \ 0
\end{bmatrix},
\quad
P_2=\begin{bmatrix}
0 & 0 & 0 & 0 \\
0 & \overline{a'} & \ol{b''} & 0 \\
0 & 0 & 0 & 0 \\
0 & 0 & 0 & 0
\end{bmatrix},
\quad
Q_1=\begin{bmatrix}
0 & 0 & 0 & 0 \\
0 & 0 & 0 & 0 \\
0 & 0 & 0 & 0 \\
\ol{c''} & 0 & 0 & \overline{d'} 
\end{bmatrix},
\quad
Q_2=\begin{bmatrix}
0 & \ 0 & 0 & 0 \\
0 & \ 0 & 0 & 0 \\
0 & -c'' & d' & 0 \\
0 & \ 0 & 0 & 0 
\end{bmatrix}.
\end{align*}
Then the product of these matrices satisfies
\begin{align*}
P_1P_2=P_1Q_2=Q_1P_2=Q_1Q_2=O,\quad P_2P_1=P_2Q_1=Q_2P_1=Q_2Q_1=O.
\end{align*} Hence $\chi(\Xi_n(l,m))$ is divided into the $\Xi_n^{(1)}(l,m)\in\mbox{M}(4,\CM)$ and $\Xi_n^{(2)}(l,m)\in\mbox{M}(4,\CM)$. Each matrix is composed by the combination of corresponding $P_i$ and $Q_i$. For example,
\begin{align*}
\chi(\Xi_3(1,2))=
\chi(PQ^2)+\chi(QPQ)+\chi(Q^2P)&=P_1Q_1^2+Q_1P_1Q_1+Q_1^2P_1+P_2Q_2^2+Q_2P_2Q_2+Q_2^2P_2
\\
&=\Xi_3^{(1)}(1,2)+\Xi_3^{(2)}(1,2),
\end{align*}
with $\Xi_3^{(1)}(1,2)=P_1Q_1^2+Q_1P_1Q_1+Q_1^2P_1$ and $\Xi_3^{(2)}(1,2)=P_2Q_2^2+Q_2P_2Q_2+Q_2^2P_2$.\\[+5pt]
Moreover each component of $\Xi_n^{(1)}(l,m)$ and $\Xi_n^{(2)}(l,m)$ corresponds to the sum of possible paths with the coin operator given as the following $U^{(1)}$ and $U^{(2)}$, respectively.
\begin{align*}
U^{(1)}=\begin{bmatrix}a' & -\overline{b'}\, \\
 c' & \ \ \overline{d'}\end{bmatrix}\in\mbox{U}(2,\CM),\hspace{1.2cm} 
U^{(2)}=\begin{bmatrix}\ \ \overline{a'} & b' \\
 -\overline{c'} & d'\end{bmatrix}\in\mbox{U}(2,\CM).
\end{align*}
By Theorem \ref{PQRS} and Lemma \ref{unitary}, we get
\begin{align*}
\Xi_n^{(1)}(l,m)=&(a')^l (\overline{d'})^m \sum_{\gamma=1}^{l\land m}\left(-\frac{|b|^2}{|a|^2}\right)^\gamma\binom{l-1}{\gamma-1}\binom{m-1}{\gamma-1}\frac{1}{|a|^2|b|^2\gamma}
\begin{bmatrix}
|a|^2|b|^2l & 0 & 0 & -(|b|^2l-\gamma)\ol{a'}b''\\
0 & 0 & 0 & 0\\
0 & 0 & 0 & 0\\
(|b|^2m-\gamma)a'\ol{b''} & 0 & 0 & |a|^2|b|^2m
\end{bmatrix},
\\[+10pt]
\Xi_n^{(2)}(l,m)=&(\overline{a'})^l (d')^m \sum_{\gamma=1}^{l\land m}\left(-\frac{|b|^2}{|a|^2}\right)^\gamma\binom{l-1}{\gamma-1}\binom{m-1}{\gamma-1}\frac{1}{|a|^2|b|^2\gamma}
\begin{bmatrix}
0 & 0 & 0 & 0\\
0 & |a|^2|b|^2l & (|b|^2l-\gamma)a'\ol{b''} & 0\\
0 & -(|b|^2m-\gamma)\ol{a'}b'' & |a|^2|b|^2m & 0\\
0 & 0 & 0 & 0\\
\end{bmatrix}.
\end{align*}
Then $\chi(\Xi_n(l,m))$ is given by
\begin{align*}
&\chi(\Xi_n(l,m))=\Xi_n^{(1)}(l,m)+\Xi_n^{(2)}(l,m)=
\\
&\sum_{\gamma=1}^{l\land m}\left(-\frac{|b|^2}{|a|^2}\right)^\gamma\binom{l-1}{\gamma-1}\binom{m-1}{\gamma-1} \frac{1}{|a|^2|b|^2\gamma}
\begin{bmatrix}
(a')^l(\overline{d'})^m& 0& 0& 0\\
0& (\overline{a'})^l(d')^m& 0& 0\\
0& 0& (\overline{a'})^l(d')^m& 0\\
0& 0& 0& (a')^l(\overline{d'})^m
\end{bmatrix}
\\[+8pt]
&\hspace{3.5cm}\begin{bmatrix}
|a|^2|b|^2l & 0 & 0 & -(|b|^2l-\gamma)\ol{a'}b''\,\\
0 & |a|^2|b|^2l & (|b|^2l-\gamma)a'\ol{b''} & 0\\
0 & -(|b|^2m-\gamma)\ol{a'}b'' & |a|^2|b|^2m & 0\\
(|b|^2m-\gamma)a'\ol{b''} & 0 & 0 & |a|^2|b|^2m
\end{bmatrix}.
\end{align*}
Noting  
$\chi(x)=\begin{bmatrix} x' & -x'' \\ \ol{x''} & \ol{x'}\end{bmatrix}$ with $x=x'+x''j\in\HM$, we have
\begin{align*}
&\Xi_n(l,m)=\sum_{\gamma=1}^{l\land m}\left(-\frac{|b|^2}{|a|^2}\right)^\gamma\binom{l-1}{\gamma-1}\binom{m-1}{\gamma-1} \frac{1}{|a|^2|b|^2\gamma}
\\
&
\hspace{5.0 cm}
\begin{bmatrix}
(a')^l(\overline{d'})^m& 0\\
0& (\overline{a'})^l(d')^m
\end{bmatrix}
\begin{bmatrix}
|a|^2|b|^2l & j(|b|^2l-\gamma)a'\ol{b''}\,\\
j(|b|^2m-\gamma)a'\ol{b''} & |a|^2|b|^2m
\end{bmatrix}.
\end{align*}
Hence
\begin{align*}
&\Psi_n(x)=\sum_{\gamma=1}^{l\land m}\left(-\frac{|b|^2}{|a|^2}\right)^\gamma\binom{l-1}{\gamma-1}\binom{m-1}{\gamma-1} \frac{1}{|a|^2|b|^2\gamma}
\\
&
\hspace{5.0cm}
\begin{bmatrix}
(a')^l(\overline{d'})^m& 0\\
0& (\overline{a'})^l(d')^m
\end{bmatrix}
\begin{bmatrix}
|a|^2|b|^2l\alpha+j(|b|^2l-\gamma)a'\ol{b''}\beta\\
j(|b|^2m-\gamma)a'\ol{b''}\alpha+|a|^2|b|^2m\beta
\end{bmatrix}.
\end{align*}
Therefore we obtain the probability distribution $P(X_n=x)$ as the same as the QW by computing $||\Psi_n(x)||^2$.

\section{Limit distribution \label{Lpd}}
This section presents the weak limit theorem of the rescaled QQW with $abcd\neq 0$. For QW, the corresponding limit theorem was given by Konno \cite{KonnoQW1}.
\begin{thm}{\rm Konno \cite{KonnoQW1}}\quad For $\displaystyle \lim_{n \to \infty}\dfrac{X_n}{n} \Rightarrow Y$, The density function $f(y)$ of Y is given by
\begin{align*}
f(y)=f(y;{}^T[\alpha,\beta])=\{1-C(a,b;\alpha,\beta)y\}f_K(y;|a|),
\end{align*}
\begin{align*}
\mbox{where}\quad
&C(a,b;\alpha,\beta)=|\alpha|^2-|\beta|^2-\frac{a\alpha\overline{b\beta}+\overline{a\alpha}b\beta}{|a|^2},\\
&f_K(y;r)=\frac{\sqrt{1-r^2}}{\pi(1-y^2)\sqrt{r^2-y^2}}I_{(-r,r)}(y)\,\quad(0<r<1).
\end{align*}
\end{thm}
\noindent Here, $I_{(-r,r)}(y)=1$, if $y\in(-r,r)$, $=0$, otherwise. Then parameter $r$ means the range of support of the limit density function.
This weak limit theorem is also obtained by Grimmett, Janson, and Scudo \cite{GJS} via the Fourier transform, which is called the {\it GJS method} in this paper.
Here we apply the GJS method to QQWs by using $\hat{\Phi}_n(\theta)\in\CM^2$.\\ 
\indent We define eigenvalues of $U(\theta)$ as $e^{i\lambda_m(\theta)}$ $(m\in\{1,2,3,4\})$ with $\lambda_m(\theta)\in[-\pi,\pi)$, since $U(\theta)$ is unitary. Put the orthonormal eigenvectors $\ket{v_m(\theta)}$ associated with $e^{i\lambda_m}$. Using the spectral decomposition of $U(\theta)$, we have
\begin{align*}
U(\theta)^n=\sum^4_{m=1}\left(e^{i\lambda_m}\right)^n\ket{v_m(\theta)}\bra{v_m(\theta)}.
\end{align*}
Set $D=i\dfrac{d}{d\theta}$, the $r$-th moment of $X_n$ can be expressed as
\begin{align}
\label{rmoment}
E(X_n^r)=\int^\pi_{-\pi}\hat{\Phi}_n^*(\theta)D^r\hat{\Phi}_n(\theta)\frac{d\theta}{2\pi},
\end{align}
since $D\hat{\Phi}_n(\theta)=x\hat{\Phi}_n(\theta)$. Moreover $D^r\hat{\Phi}_n(\theta)$ becomes
\begin{align}
\label{order}
D^r
\hat{\Phi}_n(\theta)=D^rU(\theta)^n\hat{\Phi}_0(\theta)=\sum^4_{m=1}(n)_re^{i(n-r)\lambda_m}
\left(De^{i\lambda_m}\right)^r\bra{v_m}\hat{\Phi}_0(\theta)\ket{v_m}+O(n^{r-1}),
\end{align}
where $(n)_r=n(n-1)(n-2)\cdots(n-r+1)$ and $O(f(n))$ satisfies $\displaystyle\limsup_{n \to \infty}|O(f(n))/f(n)|\leq C$ for a finite  fixed number $C$.\\
By (\ref{order}) and (\ref{rmoment}), we get
\begin{align*}
E(X_n^r)=
\sum^4_{m=1}\int^\pi_{-\pi}
(n)_r\left(\frac{De^{\lambda_m}}{e^{i\lambda_m}}\right)^r
|\bra{v_m}\hat{\Phi}_0(\theta)|^2\frac{d\theta}{2\pi}+O(n^{r-1}).
\end{align*}
Then the limit of the above equation is
\begin{align*}
\lim_{n\to\infty}E\left(\left(\dfrac{X_n}{n}\right)^r\right)=E(Y^r)=
\sum^4_{m=1}\int^\pi_{-\pi}
\left(\frac{De^{\lambda_m}}{e^{i\lambda_m}}\right)^r
|\bra{v_m}\hat{\Phi}_0(\theta)|^2\frac{d\theta}{2\pi}.
\end{align*}
Finally, by the change of the variables $\theta\to y$, that is,
\begin{align*}
\frac{De^{\lambda_m}}{e^{i\lambda_m}}=\dfrac{d}{d\theta}\lambda_m=y,
\end{align*}
we get the limit density function. The support of the limit density function is determined by this change of variables.
Moreover, the characteristic polynomial of $U(\theta)$ is
\begin{align*}
|I\lambda-U(\theta)|=
\lambda^4-2(a_0 e^{i\theta}+d_0 e^{-i\theta})\lambda^3
+2\left\{
2a_0 d_0-\Re(bc)+|a|^2\cos(2\theta))
\right\}\lambda^2
-2(d_0 e^{i\theta}+a_0 e^{-i\theta})\lambda+1
\end{align*}
and the eigenvector is
\begin{align}
\label{evec}
\ket{v(\theta)}=
\begin{bmatrix}
|b|^2(1+C_1) \\[+5pt]
-|b|^2(C_2i+C_3)\\[+5pt]
-l
\left(
\left(
\ol{b}a
\right)
',-\ol{b'}
\right)(1+C_1)
-l
\left(
\left(
\ol{b}a
\right)'',b''
\right)
(C_2i+C_3)\\[+8pt]
-l
\left(\,
\ol{
\left(
\ol{b}a
\right)
''},\ol{b''}
\right)
(1+C_1)+l
\left(\,
\ol{
\left(
\ol{b}a
\right)'},-b'
\right)(C_2i+C_3)
\end{bmatrix}
,
\end{align}
where
\begin{align*}
&C=
\frac{1}{|B|^2}
\Im
\left(
2|b|^2\sin(\lambda-\theta)a+\sin(\lambda+\theta)T+ bc\sin(2\lambda) - b\ol{d^2}c \sin(2\theta)
\right), \quad
l(x,y)=x+ye^{i(\lambda-\theta)}.
\end{align*}
Here $T=bd\ol{b}+\ol{c}dc$ and $|B|^2$ is normalized coefficient for $|C|^2=1$. Concretely,
\begin{align*}
|B^2|=&|a|^2|b|^2
\left(
\sin^2(\lambda-\theta)+\sin^2(\lambda+\theta)
\right)
-2|b|^2
\left(
a_0\sin(\lambda+\theta)+d_0\sin(\lambda-\theta)
\right)\sin (2\lambda)
\\
&-2\Re(\ol{a^2} bc)\sin(\lambda-\theta)\sin(\lambda+\theta)
+|b|^2\sin^2 (2\lambda).
\end{align*}  
The details to get the eigenvector will be written in Appendix.

\subsection{Case 5 \label{mainresult}}
In Case 5, the coin operator $U$ is given by
\begin{align*}
&U=
\begin{bmatrix}
a & b \\
c & d
\end{bmatrix}
\in{\rm U}(2,\HM)
\quad
a_0=d_0=0.
\end{align*}
Then eigenvalues of $U(\theta)$ are
\begin{align*}
Spec(U(\theta))=\{e^{i\lambda}, -e^{i\lambda}, e^{-i\lambda}, -e^{-i\lambda}\},
\end{align*}
where
\begin{align*}
\cos\lambda=\sqrt{\dfrac{1-\Re(bc)+|a|^2\cos2\theta}{2}},\quad
\sin\lambda=\sqrt{\dfrac{1+\Re(bc)-|a|^2\cos2\theta}{2}}.
\end{align*}
Here the parameter of eigenvector associated with $e^{i\lambda}$ of $U(\theta)$ in (\ref{evec}) is
\begin{align*}
&C=\frac{1}{|B|^2}\Im
\left(
2|b|^2a \sin(\lambda-\theta)+T\sin(\lambda+\theta)+bc(\sin(2\lambda)+|a|^2\sin(2\theta)\,)
\right)
,\\
&|B|^2=2|a|^2\Re(bc)\cos(2\theta)+G-2|a|^4,
\end{align*}
where $G=1+|a|^4-\Re(bc)^2$.
Then $||\ket{v(\theta)}||^2=\dfrac{4|b|^4}{|B|^2}(1+C_1)(\sin(2\lambda)+|a|^2\sin(2\theta))\sin(2\lambda)$.
Now set $y=\dfrac{d}{d\theta}\lambda$, we get
\begin{align*}
y=\frac{|a|^2\sin(2\theta)}{\sqrt{1-\Re(bc)+|a|^2\cos(2\theta)}\sqrt{1+\Re(bc)-|a|^2\cos(2\theta)}}
\end{align*}
and
\begin{align*}
\cos(2\theta)=
\begin{cases}
\dfrac{-\Re(bc)y^2+\sqrt{y^4-Gy^2+|a|^4}}{|a|^2(1-y^2)}
\quad
\left(
\frac{-\Re(bc)r^2}{|a|^2(1-r^2)}\leq \cos (2\theta) \leq 1
\right)
\\[+16pt]
\dfrac{-\Re(bc)y^2-\sqrt{y^4-Gy^2+|a|^4}}{|a|^2(1-y^2)}
\quad
\left(
-1 \leq \cos(2\theta) < \frac{-\Re(bc)r^2}{|a|^2(1-r^2)}
\right)
\end{cases}
,
\end{align*}
where $r=\sqrt{\dfrac{G-\sqrt{G^2-4|a|^4}}{2}}=\dfrac{\sqrt{(1+|a|^2)^2-\Re(bc)^2}-\sqrt{(1-|a|^2)^2-\Re(bc)^2}}{2}$.
Above mentioned method gives the limit distribution of this class of QQW as
\\
\begin{thm}\quad For 
$U=\begin{bmatrix}a & b \\ c & d\end{bmatrix}\in{\rm U}(2,\HM)$ with $a_0=d_0=0$, 
$\displaystyle \lim_{n \to \infty}\dfrac{X_n}{n} \Rightarrow Y$, the density function $f(y)$ of Y is given by
\begin{align*}
f(y)=f(y;{}^T[\alpha,\beta])=\{1-C(a,b;\alpha,\beta)y\}f_{QQW}(y;r),
\end{align*}
\begin{align*}
&\mbox{where}\\
&
f_{QQW}(y;r)=
\frac{\sqrt{2}}{2\pi(1-y^2)\sqrt{r^2-y^2}}
\frac{\sqrt{(G-2)y^2+G-2|a|^4+(1-y^2)\sqrt{G^2-4|a|^4}}}
{\sqrt{\dfrac{G+\sqrt{G^2-4|a|^4}}{2}-y^2}}
I_{(-r,r)}(y),\\
&
C(a,b;\alpha,\beta)=|\alpha|^2-|\beta|^2-\frac{a\alpha\overline{b\beta}+\overline{a\alpha}b\beta}{|a|^2},
\\
&r=\sqrt{\dfrac{G-\sqrt{G^2-4|a|^4}}{2}}=\dfrac{\sqrt{(1+|a|^2)^2-\Re(bc)^2}-\sqrt{(1-|a|^2)^2-\Re(bc)^2}}{2}\mbox{ and } G=1+|a|^4-\Re(bc)^2.
\end{align*}
\end{thm}
\noindent
Furthermore we can easily check that if $\Re(bc)=0$, then $f_{QQW}(y;r)$ corresponds to limit density function of QW, $f_{K}(y;|a|^2)$. However we remark that the support of limit distribution of such a case, $|a|^2$, is different from that of QW, $|a|$.

\par
\
\par

\begin{small}
\bibliographystyle{jplain}

\section*{Appendix}
We consider eigenvector $\ket{v(\theta)}$ of $U(\theta)$ associated with eigenvalue $e^{i\lambda}$. 
Let $U=\begin{bmatrix}a & b \\ c& d \end{bmatrix}\in{\rm U}(2,\HM)$ and
\begin{align*}
U(\theta)=\begin{bmatrix} e^{i\theta} & 0 & 0 & 0 \\ 0 & e^{i\theta} & 0 & 0 \\ 0 & 0 & e^{-i\theta} & 0 \\ 0 & 0 & 0 & e^{-i\theta}\end{bmatrix}\begin{bmatrix}\chi(a) & \chi(b) \\ \chi(c) & \chi(d)\end{bmatrix}\in{\rm U}(4,\CM),
\end{align*}
where $\chi(x)=\begin{bmatrix} x' & -x'' \\ \ol{x''} & \ol{x'} \end{bmatrix}\in{\rm M}(2,\CM)$ with $x=x'+x''j\in\HM\ $ ($x'=x_0+x_1i$, $x''=x_2+x_1i\in\CM $).
By definition, we have
\begin{align}
\label{Ueigen1}
\begin{bmatrix}
e^{i\theta} & 0 & 0 & 0 \\
0 & e^{i\theta} & 0 & 0 \\
0 & 0 & e^{-i\theta} & 0 \\
0 & 0 & 0 & e^{-i\theta}
\end{bmatrix}
\begin{bmatrix}
\chi(a) & \chi(b) \\
\chi(c) & \chi(d)
\end{bmatrix}
\ket{v(\theta)}
=e^{i\lambda}
\ket{v(\theta)}.
\end{align}
Then (\ref{Ueigen1}) is equivalent to
\begin{align}
\label{Ueigen2}
\begin{bmatrix}
\chi(a) & \chi(b) \\
\chi(c) & \chi(d)
\end{bmatrix}
\begin{bmatrix}
v_1 & -v_2 \\
\ol{v_2} & \ \  \ol{v_1} \\
v_3 & -v_4 \\
\ol{v_4} & \ \  \ol{v_3}
\end{bmatrix}
=
\begin{bmatrix}
e^{i(\lambda-\theta)}v_1 & -e^{-i(\lambda-\theta)}v_2\\
e^{i(\lambda-\theta)}\ol{v_2} & e^{-i(\lambda-\theta)}\ol{v_1}\\
e^{i(\lambda+\theta)}v_3 & -e^{-i(\lambda+\theta)}v_4\\
e^{i(\lambda+\theta)}\ol{v_4} & e^{-i(\lambda+\theta)}\ol{v_3}
\end{bmatrix}
\end{align}
with $\ket{v(\theta)}={}^T\begin{bmatrix} v_1 & \ol{v_2} & v_3 & \ol{v_4}\end{bmatrix}$.
Since, each submatrix of (\ref{Ueigen2}) is the quaternionic expansion of a quaternion derived from (\ref{Ueigen1}).
Here we can see
\begin{align*}
\begin{bmatrix}
\chi(a) & \chi(b) \\
\chi(c) & \chi(d)
\end{bmatrix}
\begin{bmatrix}
\chi(s)\\
\chi(t)
\end{bmatrix}
=
\begin{bmatrix}
\chi(\tilde{s})\\
\chi(\tilde{t})
\end{bmatrix},
\end{align*}
where 
\begin{align*}
s=v_1+v_2j,\quad \tilde{s}=e^{i(\lambda-\theta)}v_1+e^{-i(\lambda-\theta)}v_2j,\quad 
t=v_3+v_4j,\quad \tilde{t}=e^{i(\lambda+\theta)}v_3+e^{-i(\lambda+\theta)}v_4j. 
\end{align*}
So we get
\begin{align}
\label{Ueigen3}
\begin{bmatrix}
a & b\\
c  & d
\end{bmatrix}
\begin{bmatrix}
s \\
t
\end{bmatrix}
=
\begin{bmatrix}
\tilde{s}\\
\tilde{t}
\end{bmatrix}.
\end{align}
Therefore we should consider quaternions $s$, $t$ satisfying (\ref{Ueigen3}) so as to get eigenvector $\ket{v(\theta)}={}^T\begin{bmatrix} s' & \ol{s''} & t' &\ol{t''}\end{bmatrix}$. \\
Furthermore $\tilde{s}$ and $\tilde{t}$ are expressed by $s$ and $t$ as follows, respectively.
\begin{align*}
\tilde{s}=e^{i(\lambda-\theta)}v_1+e^{-i(\lambda-\theta)}v_2j=e^{i(\lambda-\theta)}v_1+v_2je^{i(\lambda-\theta)}=se^{i(\lambda-\theta)}\\
\tilde{t}=e^{i(\lambda+\theta)}v_3+e^{-i(\lambda+\theta)}v_4j=e^{i(\lambda+\theta)}v_3+v_4je^{i(\lambda+\theta)}=te^{i(\lambda+\theta)}
\end{align*}
Because of the existence of nontrivial eigenvector, we have $st\neq0$. Then (\ref{Ueigen3}) gives the following relation of $s$ and $t$.
\begin{align}
\label{As}
(1)\quad As-Bsi=0,\hspace{2cm}
(2)\quad t=\frac{\ol{b}}{|b|^2}(se^{i(\lambda-\theta)}-as),
\end{align}
where 
\begin{align*}
A=c+d\ol{b}\cos(\lambda-\theta)+\ol{b}a\cos(\lambda+\theta)-\ol{b}\cos(2\lambda),\qquad
B=-d\ol{b}\sin(\lambda-\theta)-\ol{b}a\sin(\lambda+\theta)+\ol{b}\sin2\lambda.
\end{align*}
As for (1), if $AB=0$, then (1) holds for any $s$. Hence from now on, we assume that $AB\neq 0$. Put $C=B^{-1}A$, so we have $C=sis^{-1}$ by (1). Then we see that $|C|=1$. Moreover $\Re(x)=\Re(yxy^{-1})\ (x,y\in\HM)$ implies $\Re(C)=\Re(k)=0$. By using these, we obtain
\begin{align*}
C=
\frac{1}{|B|^2}
\Im
\left(
2|b|^2\sin(\lambda-\theta)a+\sin(\lambda+\theta)T+ bc\sin(2\lambda) - b\ol{d^2}c \sin(2\theta)
\right)
\end{align*}
and
\begin{align*}
|B^2|=&|a|^2|b|^2
\left(
\sin^2(\lambda-\theta)+\sin^2(\lambda+\theta)
\right)
-2|b|^2
\left(
a_0\sin(\lambda+\theta)+d_0\sin(\lambda-\theta)
\right)\sin (2\lambda)
\\
&-2\Re(\ol{a^2} bc)\sin(\lambda_m-\theta)\sin(\lambda_m+\theta)
+|b|^2\sin^2 (2\lambda_m).
\end{align*}
Here we use the following result based on Tian (1999)\cite{Quaternion1}.
\\
\begin{thm*}
\label{1kai}
\quad For $a,b,c\in\HM$, $x\in\HM$ satisfying $ax-xb=c$ is given by\\[+18pt]
(1)
\quad  If $a_0\neq b_0$ or $|\Im(a)|\neq|\Im(b)|$, then
\begin{align*}
x&=(a^2-2b_0a+|b|^2)^{-1}(ac-c\overline{b})=\{2(a_0-b_0)a+|b|^2-|a|^2\}^{-1}(ac-c\overline{b})\\[+8pt]
&=(cb-\overline{a}c)(b^2-2a_0b+|a|^2)^{-1}=(cb-\overline{a}c)\{2(a_0-b_0)b+|b|^2-|a|^2\}^{-1}.
\end{align*}
(2)
\quad If $a_0=b_0$, $|\Im(a)|=|\Im(b)|$ and $c=0$, then
\begin{align*}
x=p-\frac{1}{|\Im(a)||\Im(b)|}\Im(a)p\Im(b) \hspace{1.5cm} (p\in\HM).
\end{align*}
\\
(3)
\quad If $a_0=b_0$, $|\Im(a)|=|\Im(b)|$ and $c\neq0$, then
\begin{align*}
x=\frac{1}{4|a|^2}(cb-ac)-\frac{1}{|\Im(a)||\Im(b)|}\Im(a)p\Im(b) \hspace{1.5cm} (p\in\HM).
\end{align*}
\end{thm*}
\noindent By applying (2) in this theorem for $Cs+sk=0$, we get
\begin{align*}
s=p-Cpi \hspace{1cm} (p\in\HM/\{0\}).
\end{align*}
Then we have simplex and perplex parts of $s$ and $t$:
\begin{align*}
s=&
\left\{
(1+C_1)p'+(C_2i-C_3)\ol{p''}
\right\}
+
\ol
{
\left\{
-(C_2i+C_3)p'+(1-C_1)\ol{p''}
\right\}
}j,
\\
t=&
\frac{1}{|b|^2}
{\bigg [}
\left\{
-l
\left(
\left(
\ol{b}a
\right)
',-\ol{b'}
\right)(1+C_1)
-l
\left(
\left(
\ol{b}a
\right)'',b''
\right)
(C_2i+C_3)
\right\}
p'
\\
&
\hspace{2.0 cm}
+
\left\{
l
\left(
\left(
\ol{b}a
\right)
'', b''
\right)(1-C_1)
-l
\left(
\left(
\ol{b}a
\right)', -\ol{b'}
\right)
(C_2i-C_3)
\right\}
\ol{p''}\ 
{\bigg ]}
\\
&
+
\frac{1}{|b|^2}
{\bigg [}\ 
\ol{
\left\{
-l
\left(\,
\ol{
\left(
\ol{b}a
\right)
''},\ol{b''}
\right)
(1+C_1)+l
\left(\,
\ol{
\left(
\ol{b}a
\right)'},-b'
\right)(C_2i+C_3)
\right\}p'
}
\\
&
\hspace{2.0 cm}
+
\ol{
\left\{
-l
\left(\,
\ol{
\left(
\ol{b}a
\right)
'},-b'
\right)
(1-C_1)-l
\left(\,
\ol{
\left(
\ol{b}a
\right)''},\ol{b''}
\right)(C_2i-C_3)
\right\}
\ol{p''}
}
\ {\bigg ]}
j.
\end{align*}
where $p=p'+p''j$. Therefore eigenvector $\ket{v(\theta)}={}^T\begin{bmatrix} s' & \ol{s''} & t' & \ol{t''}\end{bmatrix}$ is expressed as $C,\  l(x,y)=x+ye^{i(\lambda-\theta)}$, and $p$ in the following:
\begin{align*}
\ket{v_m(\theta)}=
\frac{1}{|b|^2}
\begin{bmatrix}
|b|^2(1+C_1) \\[+5pt]
-|b|^2(C_2i+C_3)\\[+5pt]
-l
\left(
\left(
\ol{b}a
\right)
',-\ol{b'}
\right)(1+C_1)
-l
\left(
\left(
\ol{b}a
\right)'',b''
\right)
(C_2i+C_3)\\[+8pt]
-l
\left(\,
\ol{
\left(
\ol{b}a
\right)
''},\ol{b''}
\right)
(1+C_1)+l
\left(\,
\ol{
\left(
\ol{b}a
\right)'},-b'
\right)(C_2i+C_3)
\end{bmatrix}
p'
\\
+
\frac{1}{|b|^2}
\begin{bmatrix}
|b|^2(C_2i-C_3) \\[+5pt]
|b|^2(1-C_1)\\[+5pt]
l
\left(
\left(
\ol{b}a
\right)
'', b''
\right)(1-C_1)
-l
\left(
\left(
\ol{b}a
\right)', -\ol{b'}
\right)
(C_2i-C_3)\\[+8pt]
-l
\left(\,
\ol{
\left(
\ol{b}a
\right)
'},-b'
\right)
(1-C_1)-l
\left(\,
\ol{
\left(
\ol{b}a
\right)''},\ol{b''}
\right)(C_2i-C_3)
\end{bmatrix}
\ol{p''}
.
\end{align*}
Let $p'=|b|^2$ and $\ol{p''}=0$, so we get the eigenvector of $U(\theta)$ in (\ref{evec}).

\end{small}

\end{document}